\documentclass{llncs}
\usepackage{graphicx}
\usepackage[utf8]{inputenc}
\usepackage[fleqn]{amsmath}
\usepackage{amssymb}
\usepackage[lined,boxed]{algorithm2e}
\usepackage{epstopdf}
\allowdisplaybreaks

\DeclareMathOperator*{\argmax}{arg\,max}

\def \fs {{\cal H}}

\def \ew {c}
\def \vw {f}
\def\R{\mathbb{R}}
\def\gc{\ensuremath{{\cal C}_G}\xspace}
\def\gcmc{\ensuremath{\mathsf{GCMC}}\xspace}

\def\sse{\subseteq}

\def \SA {Sherali-Adams }
\newtheorem{define}{Definition}
\newtheorem{cl}{Claim}
\newtheorem{observation}{Observation}
\newtheorem{assume}{Assumption}
\def\eod{\vrule height 6pt width 5pt depth 0pt}

\title{Max-Cut under Graph Constraints}

\author{Jon Lee\thanks{Research of J. Lee was partially supported by
NSF grant CMMI–1160915 and ONR grant N00014-14-1-0315.} \and Viswanath Nagarajan \and Xiangkun Shen }
\institute{IOE Dept., University of Michigan, Ann Arbor, MI 48109, USA.}
\date{}

\begin{document}

\maketitle
\begin{abstract}
An instance of the graph-constrained max-cut (\gcmc) problem consists of (i) an undirected graph $G=(V,E)$ and (ii) edge-weights $c:{V\choose 2} \rightarrow \mathbb{R}_+$ on a complete undirected graph. The objective is to find a subset $S \subseteq V$ of vertices satisfying some graph-based constraint in $G$ that maximizes the weight $\sum_{u\in S, v\not\in S} c_{uv}$ of edges in the cut $(S,V\setminus S)$. The types of graph constraints we can handle include independent set, vertex cover, dominating set and connectivity. Our main results are for the case when $G$ is a graph with bounded treewidth, where we obtain a  $\frac12$-approximation algorithm. Our algorithm uses an  LP relaxation based on the Sherali-Adams hierarchy. It can handle any graph constraint for which there is a (certain type of) dynamic program  that exactly optimizes linear objectives.

\smallskip
Using known decomposition results, these imply essentially the same approximation ratio for \gcmc under constraints such as independent set, dominating set and connectivity on a planar graph $G$ (more generally for bounded-genus or excluded-minor graphs).
\end{abstract}

\section{Introduction}\label{sec:intro}

The max-cut problem is an extensively studied combinatorial-optimization problem. Given an undirected edge-weighted graph, the goal   is to find a subset $S \subseteq V$ of vertices that maximizes the weight of edges in the cut $(S,V\setminus S)$. Max-cut has a 0.878-approximation algorithm~\cite{GW95} which is known to be best-possible assuming the ``unique games conjecture''~\cite{KMO07}. It also has a number of practical applications, e.g., in circuit layout, statistical physics and clustering.

In some applications, one needs to solve the max-cut problem under additional constraints on the subset $S$. Consider for example, the following clustering problem. The input is an undirected graph $G=(V,E)$ representing, say, a social network (vertices $V$ denote users and edges $E$ denote connections between users), and a weight function $\ew:{V\choose 2} \rightarrow \R_+$ representing,   a dissimilarity measure between pairs of users. The goal  is to find a subset $S\sse V$ of users that are connected in $G$ while maximizing the weight of edges in the cut $(S,V\setminus S)$. This corresponds to finding a cluster of connected users that is as different as possible from its complement set. This ``connected max-cut'' problem also arises in image segmentation applications~\cite{VKR08,HKMPS15}.

Designing algorithms for constrained versions of max-cut is also interesting from a theoretical standpoint. For max-cut under certain types of constraints (such as cardinality or matroid constraints) good approximation algorithms are known, e.g.,~\cite{AS99,AHS01}. In fact, many of these results have since been extended to the more general setting of submodular objectives~\cite{FNS11,CVZ14}. However, not much is known for max-cut under ``graph-based'' constraints as in the example above.

In this paper, we study a large class of graph-constrained max-cut problems and present unified approximation algorithms for them. Our results require that the constraint be defined on a graph $G$ of bounded treewidth. (Treewidth is a measure of how similar a graph is to a tree structure --- see \S\ref{sec:prelim} for definitions.) We note however that for a number of constraints (including the connectivity example above),  we can combine our algorithm with known decomposition results~\cite{DHK05,DHK11} to obtain essentially the same approximation ratios when the constraint graph $G$ is planar/bounded-genus/excluded-minor.

\paragraph{Problem definition.} The input to the {\em graph-constrained max-cut} (\gcmc) problem consists of (i) an $n$-vertex undirected graph $G=(V,E)$ which implicitly specifies a collection ${\cal C}_G\subseteq 2^V$ of feasible vertex subsets, and (ii) (symmetric) edge-weights $\ew:{V\choose 2} \rightarrow \R_+$. The \gcmc problem is then as follows:
\begin{equation}\label{eq:gcmc-defn}
\max_{S\in \gc} \quad \sum_{u\in S, v\not\in S} \ew(u,v).
\end{equation}
In this paper, we assume that the constraint graph $G$ has bounded treewidth. We also assume that the graph constraint \gc admits an exact dynamic program for optimizing a linear objective, i.e. for:
\begin{equation}\label{eq:LinObj-defn}
\max_{S\in \gc} \quad \sum_{u\in S} \vw(u),\qquad \mbox{where $\vw:V\rightarrow \R$ is any given vertex weights}.
\end{equation}

Note that the \gcmc objective~\eqref{eq:gcmc-defn} is a quadratic function of the solution $S$, whereas our assumption~\eqref{eq:LinObj-defn} involves a {\em linear} function of the solution $S$. See \S\ref{sec:prelim} for more precise definitions/assumptions.

\subsection{Our Results and Techniques}\label{subsec:results}
Our main result can be stated informally as follows.

\begin{theorem}[{\small \gcmc result --- informal}] \label{thm:main1}
Consider any instance of the \gcmc problem on a bounded-treewidth graph $G=(V,E)$. Suppose there is an exact dynamic program for optimizing any linear function subject to constraint \gc. Then we obtain a $\frac12$-approximation algorithm for \gcmc.
\end{theorem}

This algorithm uses a linear-programming relaxation for \gcmc based on the dynamic program (for linear objectives) which is further strengthened via the Sherali-Adams LP hierarchy. The resulting LP has polynomial size whenever the number of dynamic program states associated with a single tree-decomposition node is constant (see \S\ref{sec:prelim} for the formal definition).\footnote{For other polynomial time dynamic programs, the LP has {\em quasi-polynomial} size.} The rounding algorithm is a natural top-down procedure that randomly chooses a ``state'' for each tree-decomposition node using the LP's probability distribution conditional on the choices at its ancestor nodes. The final solution is obtained by combining the chosen states at each  tree-decomposition node, which is guaranteed to satisfy constraint \gc due to properties of the dynamic program. We note that the choice of variables in the Sherali-Adams LP as well as the rounding algorithm are similar to those used in~\cite{GTW13} for the sparsest cut problem on bounded-treewidth graphs. An important difference in our result is that we apply the Sherali-Adams hierarchy to a non-standard LP that is defined using the dynamic program for linear objectives. (If we were to apply Sherali-Adams to the standard LP, then it is unclear how to enforce the constraint \gc during the rounding algorithm.) Another difference is that our rounding algorithm needs to make a correlated choice in selecting the states of sibling nodes in order to satisfy constraint \gc --- this causes the number of variables in the Sherali-Adams LP to increase, but it still remains polynomial since the tree-decomposition has constant degree.

\smallskip
The requirement in Theorem~\ref{thm:main1} on the graph constraint \gc is satisfied by several interesting constraints and thus we obtain approximation algorithms for all these \gcmc problems. See Section~\ref{sec:appln} for details.
\begin{theorem}[Applications] \label{thm:main2} There is a $\frac12$-approximation algorithm for \gcmc under the following constraints in a bounded-treewidth graph: independent set,  vertex cover, dominating set,  connectivity.
\end{theorem}

We note that many other constraints such as precedence, connected dominating set, and triangle matching also satisfy our requirement. In the interest of space, we only present details for the constraints mentioned in Theorem~\ref{thm:main2}. We also note that for some of these constraints (e.g., independent set) one can come up with a problem specific algorithm where the approximation ratio depends on the treewidth $k$. Our result is stronger since the algorithm is more general, and the ratio is independent of $k$.

For many of the constraints above, we can use known decomposition results~\cite{DHK05,DHK11} to obtain approximation algorithms for \gcmc when the constraint graph has bounded genus or excludes some fixed minor (e.g., planar graphs).
\begin{corollary}
There is a $(\frac12-\epsilon)$-approximation algorithm for \gcmc under the following constraints in an excluded-minor graph: independent set,  vertex cover, dominating set. Here $\epsilon>0$ is a fixed constant.
\label{cor:em}
\end{corollary}
\begin{corollary}
There is a $(\frac12-\epsilon)$-approximation algorithm for connected max-cut in a bounded-genus graph. Here $\epsilon>0$ is a fixed constant. \label{cor:bg}
\end{corollary}

\noindent Our approach can also handle other types of objectives. If $g:2^V\rightarrow \R_+$ is
the sum of a polynomial number of functions   each of which is monotone, submodular and defined on a constant-size subset of $V$, then we obtain a $(1-\frac{1}{e})$-approximation algorithm for the problem of maximizing $g(S)$ subject to $S\in \gc$. The graph constraint \gc is as above.\footnote{This setting  is interesting only for constraints such as independent set, triangle matching and precedence that are not ``upward closed''.} The main idea is to use the correlation gap of monotone submodular function.\cite{ADSY12,CVZ14}

\subsection{Related Work}
For the basic undirected max-cut problem, there is an elegant 0.878-approximation algorithm~\cite{GW95} via semidefinite programming.  This is also the best one can hope for, assuming the {\em unique games conjecture}~\cite{KMO07}.

Most of the prior work on constrained max-cut  has focused on cardinality, matroid and knapsack constraints~\cite{AS99,AHS01,FNS11,CVZ14,LMNS10,LSV10}.
Constant-factor approximation algorithms are known for max-cut under the intersection of any constant number of such constraints ---  these results hold in the substantially more general setting of non-negative submodular functions. The main techniques used here are local search and the multilinear extension~\cite{CCPV11} of submodular functions. These results made crucial use of certain exchange properties  of the underlying constraints, which are not true for graph-based constraints that we consider. 

Closer to our setting, a version of the connected max-cut problem was studied recently in~\cite{HKMPS15}, where the connectivity constraint as well as the weight function were defined on the {\em same} graph $G$. The authors obtained an $O(\log n)$-approximation algorithm for general graphs, and an exact algorithm on bounded-treewidth graphs (which implied a PTAS for bounded-genus graphs); their algorithms relied heavily on the uniformity of the constraint/weight graphs. In contrast, we consider the connected max-cut problem where the connectivity constraint and the weight function are {\em unrelated}; in particular, our problem generalizes max-cut even when $G$ is a trivial graph (e.g., a star).
Moreover, our algorithms work for a much wider class of constraints. We note however that our results require  graph $G$ to have bounded treewidth --- this is also necessary since some of the constraints we consider (e.g., independent set) are inapproximable in general graphs. (For connected max-cut itself, obtaining a non-trivial approximation ratio when $G$ is a general graph remains an open question.)

In terms of techniques, the closest work to ours is~\cite{GTW13}. We use ideas
from~\cite{GTW13} in formulating the (polynomial size) Sherali-Adams LP as well as in the rounding algorithm. There are important differences too, as discussed in \S\ref{subsec:results}.

Finally, our result adds to a somewhat small list~\cite{BO04,MM09,BCG09,BLNZ15,GTW13,FKLST14} of algorithmic results based on the Sherali-Adams~\cite{SA90} LP hierarchy. We are not aware of a more direct approach to obtain a constant-factor approximation algorithm even for connected max-cut when the constraint graph $G$ is a tree.

\section{Preliminaries}\label{sec:prelim}

\noindent {\em Basic definitions.} For an undirected complete graph on vertices $V$ and subset $S\subseteq V$, let $\delta S$ be the set of edges with exactly one end-point in $S$. For any weight function $\ew: {V\choose 2} \rightarrow \R_+$ and subset $F\sse {V\choose 2}$, we use $\ew(F):=\sum_{e\in F} \ew_{e}$.

\noindent {\em Tree Decomposition.}
Given an undirected graph $G=(V,E)$, this  consists of a tree ${\cal T}=(I,F)$ and a collection of vertex subsets $\{X_i\subseteq V\}_{i\in I}$ such that:
\begin{itemize}
\item for each $v\in V$, the nodes $\{i \in I: v\in X_i\}$ are connected in ${\cal T}$, and
\item for each edge $(u,v)\in E$, there is some node $i\in I$ with $u,v\in X_i$.
\end{itemize}

The width of such a tree-decomposition is $\max_{i\in I}(|X_i|-1)$, and the treewidth of $G$ is the smallest width of any tree-decomposition for $G$.

We will work with ``rooted'' tree-decompositions that also specify a root node $r\in I$. The depth $d$ of such a tree-decomposition is the length of the longest root-leaf path in ${\cal T}$. The depth of any node $i\in I$ is the length of the $r-i$ path in ${\cal T}$.

For any $i\in I$, the set $V_i$ denotes all the vertices at or below node $i$, that is 
\vspace{-1mm}$$V_i\quad :=\quad \cup_{k \in {\cal T}_i}\,\,   X_k, \quad \mbox{where }{\cal T}_i=\{ k\in I : k \mbox{ in subtree of ${\cal T}$ rooted at }i\}.\vspace{-1mm}$$

The following known result provides a convenient representation of ${\cal T}$.

\begin{theorem}[Balanced Tree Decomposition~\cite{B89}]
Let $G=(V,E)$ be a graph of treewidth $k$. Then $G$ has a rooted tree-decomposition $({\cal T}=(I,F), \{X_i|i\in I\})$ where ${\cal T}$ is a binary tree of depth $2\lceil \log_\frac{5}{4} (2|V|)\rceil$ and treewidth at most $3k+2$. This tree-decomposition can be found in $O(|V|)$ time.
\label{thm:treedecomp}
\end{theorem}

\paragraph{Dynamic program for linear objectives.} We assume that the constraint \gc admits an exact dynamic programming (DP) algorithm for optimizing linear objectives, i.e.  for the problem~\eqref{eq:LinObj-defn}.
There is some additional notation that is needed to formally describe the DP: this is necessary due to the generality of our results.

\begin{define}[{\small DP}]\label{def:DP}
With any tree-decomposition {\small $({\cal T}=(I,F), \{X_i|i\in I\})$}, we associate the following:
\begin{enumerate}
\item For each node $i\in I$, there is a state space $\Sigma_i$.
\item For each node $i\in I$ and $\sigma\in\Sigma_i$, there is a collection $\fs_{i,\sigma}\subseteq 2^{V_i}$ of  subsets.
\item For each node $i\in I$, its children nodes $\{j,j'\}$ and $\sigma\in\Sigma_i$, there is a collection ${\cal F}_{i,\sigma}\subseteq \Sigma_j\times \Sigma_{j'}$ of valid combinations of children states.
\end{enumerate}
\end{define}

\begin{assume}[{\small Linear objective DP for \gc}]\label{assm:DP}
Let $({\cal T}=(I,F), \{X_i|i\in I\})$ be any  tree-decomposition. Then there exist $\Sigma_i$, ${\cal F}_{i,\sigma}$ and $\fs_{i,\sigma}$ (see Definition~\ref{def:DP}) that satisfy the following:
\begin{enumerate}
\item \emph{(bounded state space)} $\Sigma_i$ and ${\cal F}_{i,\sigma}$ are all bounded by constant, that is, $\max_i{|\Sigma_i|}=t$ and $\max_{i,\sigma}{|{\cal F}_{i,\sigma}|}=p$, where $t, p = O(1)$.
\item \emph{(required state)} For each $i\in I$ and $\sigma\in\Sigma_i$, the intersection with $X_i$ of every set in $\fs_{i,\sigma}$ is the same, denoted $X_{i,\sigma}$, that is $S\cap X_i=X_{i,\sigma}$ for all $S\in\fs_{i,\sigma}$.
\item \emph{(subproblem)} For each non-leaf node $i\in I$ with children $\{j,j'\}$ and $\sigma\in\Sigma_i$,
$$\fs_{i,\sigma}\quad =\quad \left\{X_{i,\sigma}\cup S_j\cup S_{j'} \,\,:\,\,S_j\in\fs_{j,w_j}, \, S_{j'} \in\fs_{j',w_{j'}}, \, (w_j,w_{j'}) \in{\cal F}_{i,\sigma}\right\}.$$
By condition~2, for any leaf $\ell\in I$ and $\sigma\in\Sigma_\ell$, we have $\fs_{\ell,\sigma}=\{X_{\ell,\sigma}\}$ or $\emptyset$.
\item \emph{(cover all constraints)} At the root node $r$, we have ${\cal C}_G=\bigcup_{\sigma\in\Sigma_r}{\fs_{r,\sigma}}$.
\end{enumerate}
\end{assume}

The most restrictive assumption is the first condition. To the best of our knowledge, all natural constraints that admit a dynamic program on bounded-treewidth graphs (for linear objectives) satisfy conditions 2-4. Even in cases when condition 1 is not true, a relaxed version holds  (where $t$ and $p$ are polynomial), and our approach gives a {\em quasi-polynomial} time $\frac12$-approximation algorithm.

\paragraph{Example:} Here we outline how independent set satisfies the above requirements.

\begin{itemize}
\item The state space of each node $i\in I$ consists of all  independent subsets of $X_i$.

\item The subsets $\fs_{i,\sigma}$  consist of all independent subsets $S\sse V_i$ with $S\cap X_i=\sigma$.
\item The valid combinations ${\cal F}_{i,\sigma}$ consist of  all tuples $(w_j,w_{j'})$ where the child states $w_j$ and $w_{j'}$ are ``consistent'' with state $\sigma$ at node $i$.

\end{itemize}
A formal proof of why the independent-set constraint satisfies Assumption~\ref{assm:DP} appears in Section~\ref{sec:appln}.
There, we also discuss a number of other graph constraints satisfying our assumption.

The following result follows from Assumption~\ref{assm:DP}.

\begin{cl}\label{cl:DP-sol-state}
For any $S\in \gc$, there is a collection $\{b(i)\in \Sigma_i\}_{i\in I}$ such that:
\begin{itemize}
\item for each node $i\in I$ with children $j$ and $j'$, $(b(j),b(j')) \in {\cal F}_{i,b(i)}$,
\item for each leaf $\ell$ we have $\fs_{\ell,b(\ell)}\ne \emptyset$, and
\item $S=\bigcup_{i\in I} X_{i,b(i)}$.
\end{itemize}
Moreover, for any vertex $u\in V$, if $\overline{u}\in I$ denotes the highest node containing $u$ then $u\in S \iff u\in X_{\overline{u},b(\overline{u})}$.
\end{cl}
\begin{proof}
We define the states $b(i)$ in a top-down manner; we will also define an associated subset $B_i\in \fs_{i,b(i)}$ at each node $i$. At the root, we set $b(r)=\sigma$ such that $S\in \fs_{r,\sigma}$: this is well-defined by Assumption~\ref{assm:DP}(4). We also set $B_r=S$. Having set $b(i)$ and $B_i\in \fs_{i,b(i)}$ for any node $i\in I$ with children $\{j,j'\}$, we use Assumption~\ref{assm:DP}(3)  to write
$$B_i=X_{i,b(i)} \cup S_j\cup S_{j'} \quad \mbox{where } S_j\in \fs_{j,w_j} ,\, S_{j'}\in \fs_{j',w_{j'}} \mbox{ and } (w_j,w_{j'}) \in {\cal F}_{i,b(i)}.$$
Then we set $b(j)=w_j$ and $B_j=S_j$ for all the children $J$ of node $i$. It is now easy to verify the first three conditions in the claim.

Since $S=\bigcup_{i\in I} X_{i,b(i)}$, it is clear that $u\in X_{\overline{u},b(\overline{u})} \implies u\in S$. In the other direction, suppose $u\not\in X_{\overline{u},b(\overline{u})}$: we will show $u\not\in S$. Since $\overline{u}$ is the highest node containing $u$, it suffices to show that $u\not\in  B_{\overline{u}}$ above. But this follows directly from Assumption~\ref{assm:DP}(2) since $B_{\overline{u}}\in \fs_{\overline{u},b(\overline{u})}$, $u\in X_{\overline{u}}$ and  $u\not\in X_{\overline{u},b(\overline{u})}$.
\hfill \eod
\end{proof}

\paragraph{\SA LP hierarchy.}
This is  one of the several ``lift-and-project'' procedures that, given a $\{0,1\}$ integer program, produces systematically a sequence of increasingly tighter convex relaxations. The \SA procedure~\cite{SA90} involves generating stronger {\em LP relaxations} by adding new variables and constraints. The $r^{th}$ round of this procedure has a variable $y(S)$ for every subset $S$ of at most $r$ variables in the original integer program --- the new variable $y(S)$ corresponds to the joint event that all the original variables in $S$ are one.

\section{Approximation Algorithm for \gcmc}\label{sec:sym}\label{sec:alg}

In this section, we prove:
\begin{theorem}\label{thm:sym}
Consider any instance of the \gcmc problem on a bounded-treewidth graph $G=(V,E)$. If the graph constraint \gc satisfies Assumption~\ref{assm:DP} then we obtain a

$\frac12$-approximation algorithm.
\end{theorem}
\paragraph{Algorithm outline:}
We start with a balanced tree-decomposition ${\cal T}$ of graph $G$, as given in Theorem~\ref{thm:treedecomp}; recall the associated definitions from \S\ref{sec:prelim}. Then we formulate an LP relaxation of the problem using Assumption~\ref{assm:DP} (i.e. the dynamic program for linear objectives) and further strengthened by applying the \SA operator. Finally we use a natural top-down rounding that relies on Assumption~\ref{assm:DP} and the \SA constraints. 

\subsection{Linear Program}

We start with some additional notation related to the tree-decomposition  ${\cal T}$ (from Theorem~\ref{thm:treedecomp}) and our dynamic program assumption (Assumption~\ref{assm:DP}).
\begin{itemize}
\item For any node $i\in I$, $T_i$ is the set of nodes on the $r-i$ path along with the children of all nodes except $i$ on this path. See also Figure~\ref{fig:SA-variable}.
\item ${\cal P}$ is the collection of all node subsets $J$ such that $J\sse T_{\ell_1} \cup T_{\ell_2}$ for some pair of leaf-nodes $\ell_1,\ell_2$. See also Figure~\ref{fig:SA-variable}.
\item $s(i)\in \Sigma_i$ denotes a state at node $i$. Moreover, for any subset of nodes $N\sse I$, we use the shorthand  $s(N):=\{s(k):k\in N\}$.
\item $\bar{u}\in I$ denotes the highest tree-decomposition node containing vertex $u$.
\end{itemize}

\begin{figure}[htb]
  \centering
  \includegraphics[height=0.45\textwidth]{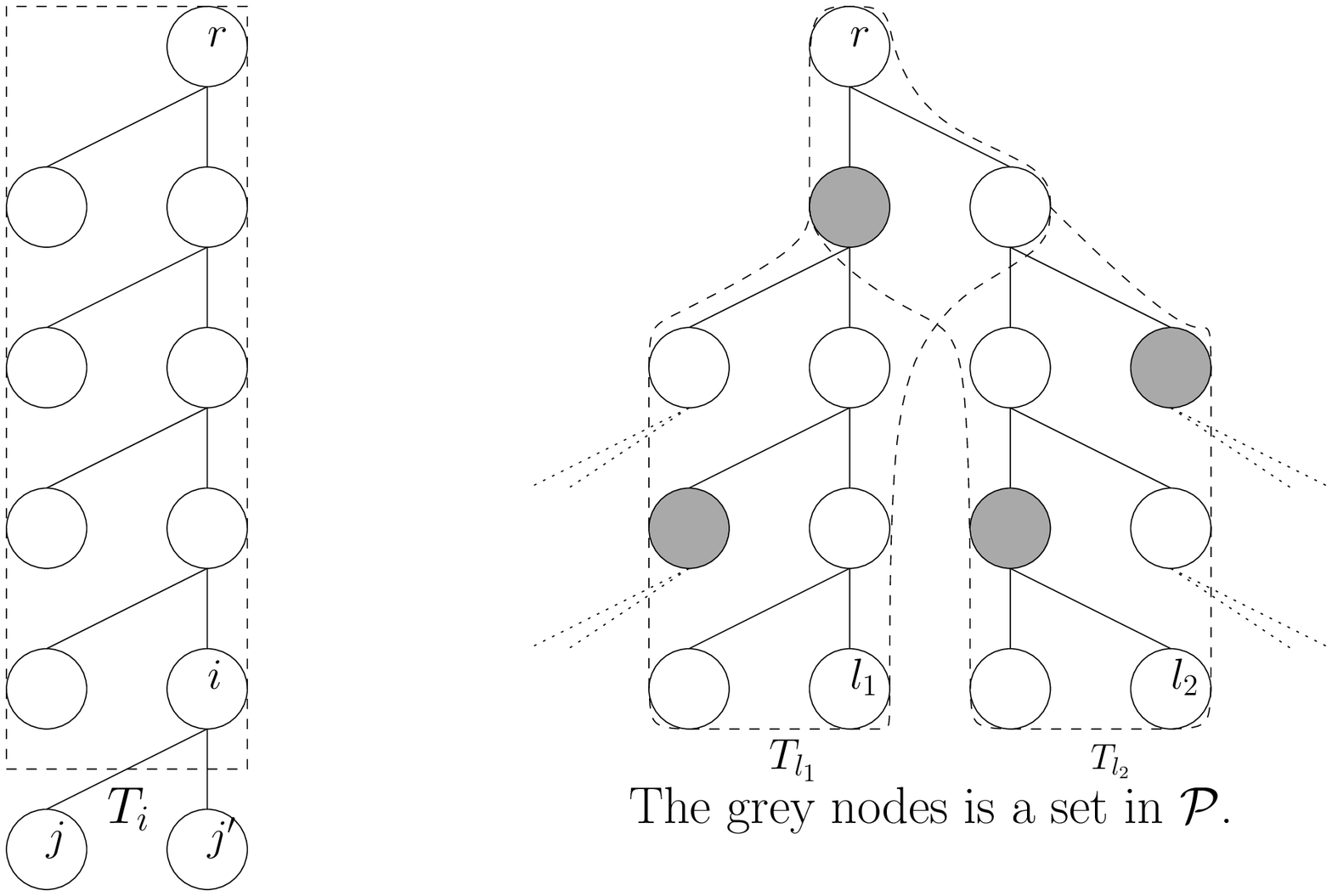}

\caption{Examples of (i) a set $T_i$ and (ii) a set in ${\cal P}$.\label{fig:SA-variable}}
\label{fig:partition}
\end{figure}

The variables in our LP are $y(s(N))$ for all $\{s(k)\in \Sigma_k\}_{k\in N}$ and $N\in {\cal P}$.  Variable $y(s(N))$ corresponds to the joint event that the solution (in \gc) ``induces'' state $s(k)$ (in terms of Assumption~\ref{assm:DP}) at each node $k\in N$.

We also use variables $z_{uv}$ defined in constraint~\eqref{cons:obj} that measure the probability of an edge $(u,v)$ being cut. Constraints~\eqref{cons:SA} are the Sherali-Adams constraints that enforce consistency among the $y$ variables. Constraints~\eqref{cons:r}-\eqref{cons:DP-leaf} are from the dynamic program (Assumption~\ref{assm:DP}) and require valid state selections. 
\begin{flalign}
&\mbox{maximize } \sum_{\{u,v\}\in{V \choose 2}}{c_{uv}z_{uv}}\tag{LP}\label{LP}&\\
&z_{uv}\,\, =\,\, \sum_{\substack{s(\bar{u})\in\Sigma_{\bar{u}} \\ u\in X_{\bar{u},s(\bar{u})}}}\sum_{\substack{s(\bar{v})\in\Sigma_{\bar{v}} \\ v\not\in X_{\bar{v},s(\bar{v})}}}{y(s(\{\bar{u},\bar{v}\}))}+\sum_{\substack{s(\bar{u})\in\Sigma_{\bar{u}} \\ u\not\in X_{\bar{u},s(\bar{u})}}}\sum_{\substack{s(\bar{v})\in\Sigma_{\bar{v}} \\ v\in X_{\bar{v},s(\bar{v})}}}{y(s(\{\bar{u},\bar{v}\}))}, & \notag \\
 &\qquad \qquad \qquad \qquad \qquad \qquad \qquad \qquad \qquad \qquad \qquad \qquad \forall \{u,v\}\in{V \choose 2};\label{cons:obj}&\\
&y(s(N))\,\, =\,\, \sum_{s(i)\in\Sigma_i}{y(s(N\cup\{i\}))}, \qquad \forall N\in {\cal P},\,  i\notin N\,:\, N\cup \{i\}\in {\cal P};\label{cons:SA}&\\ 
&\sum_{s(r)\in\Sigma_r}y(s(r))=1;&\label{cons:r}\\
&y(s(\{i,j,j'\}))\,=\,0, \qquad \qquad \quad \forall i\in I, \,  s(i)\in \Sigma_i,\, (s(j),s(j'))\notin {\cal F}_{i,s(i)}; \label{cons:DP}&\\ 
&y(s(\ell))\,=\,0, \qquad \qquad \qquad \qquad \qquad \quad \,\,\, \forall \ell\in I, s(\ell)\in \Sigma_\ell : \fs_{\ell,s(\ell)} = \emptyset ; \label{cons:DP-leaf}&\\ 
&0 \,\,\le \,\,y(s(N))\,\,\le \,\, 1,\qquad \qquad \qquad \qquad \qquad \forall N\in {\cal P},\, \{s(k)\in\Sigma_k\}_{k\in N}. &\label{cons:01}
\end{flalign}

\begin{cl}\label{cl:sym-LP-distr}
For any node $i\in I$ with children $j,j'$ and $s(k)\in\Sigma_k$ for all $k\in T_i$,
\begin{equation}
y(s(T_i))\,\,=\,\,\sum_{s(j)\in\Sigma_j}\,\, \sum_{s(j')\in\Sigma_{j'}}{y(s(T_i\cup\{j,j'\})}.
\end{equation}\label{lm:SADP}
\end{cl}
\begin{proof}
Note that $T_i\cup\{j,j'\}\sse T_\ell$ for any leaf node $\ell$ in the subtree below $i$. So $T_i\cup\{j,j'\}\in{\cal P}$ and the variables $y(s(T_i\cup\{j,j'\})$ are well-defined. The claim now follows by two applications of constraint \eqref{cons:SA}. 
\hfill \eod
\end{proof}

In constraint~\eqref{cons:DP}, we use $j$ and $j'$ to denote the two children of node $i\in I$.

\subsection{The Rounding Algorithm}
We start with the root node $r\in I$. Here $\{y(s( r))\,:\, s(r)\in\Sigma_r\}$ defines a probability distribution over the states of $r$. We sample a state $a(r)\in\Sigma_r$ from this distribution. Then we continue top-down: given the chosen state $a(i)$ of any node $i$, we sample states for both children of $i$ {\em simultaneously} from their joint distribution given at node $i$.

\begin{algorithm}[H]
\LinesNumbered
\SetKwInOut{Input}{Input}\SetKwInOut{Output}{Output}
\SetKwFor{Do}{Do}{\string:}{end}
 \Input{Optimal solution of \ref{LP}.}
 \Output{A vertex set in ${\cal C}_G$.}
\label{step:sym-round-sample} Sample a state $a(r)$ at the root node by distribution $y(s( r))$\;
 \Do{process all nodes $i$ in ${\cal T}$ in order of increasing depth}{
  Sample states $a(j),a(j')$ for the children of node $i$ by joint distribution
\begin{equation}\Pr[a(j)=s(j)\mbox{ and }a(j')=s(j')] \,\,= \,\, \frac{y(s(T_i\cup\{j,j'\}))}{y(s(T_i))},\label{eq:sym-SA-round}\end{equation}
 where   $s(T_i)=a(T_i)$.
 }
  \Do{process all nodes $i$ in ${\cal T}$ in order of decreasing depth}{
 $R_i =X_{i,a(i)} \cup R_{j}\cup R_{j'}$ where $j,j'$ are the children of $i$.
 }
$R=R_r$\;
\Return $R$.
\caption{Rounding Algorithm for \ref{LP}}
\end{algorithm}
\subsection{Algorithm Analysis}
\begin{lemma}
 \eqref{LP} is a valid relaxation of \gcmc.
\end{lemma}
\begin{proof}
Let $S\in \gc$ be any feasible solution to the \gcmc instance. Let $\{b(i)\}_{i\in I}$ denote the states given by Claim~\ref{cl:DP-sol-state} corresponding to $S$. For any subset $N\in {\cal P}$ of nodes, and for all $\{s(i)\in \Sigma_i\}_{i\in N}$, set
$$y(s(N)) = \left\{
\begin{array}{ll}
1, & \mbox{ if $s(i)=b(i)$ for all }i\in N;\\
0, & \mbox{ otherwise.}
\end{array}\right.
$$
It is easy to see that constraints~\eqref{cons:SA} and~\eqref{cons:01} are satisfied. By the first two properties in Claim~\ref{cl:DP-sol-state}, it follows that constraints~\eqref{cons:DP} and~\eqref{cons:DP-leaf} are also satisfied. The last property in Claim~\ref{cl:DP-sol-state} implies that $u\in S \iff u\in X_{\overline{u},b(\overline{u})}$ for any vertex $u\in V$. So any edge $\{u,v\}$ is cut exactly when one of the following occurs:
\begin{itemize}
\item $u\in X_{\overline{u},b(\overline{u})}$ and $v\not\in X_{\overline{v},b(\overline{v})}$;
\item $u\not\in X_{\overline{u},b(\overline{u})}$ and $v \in X_{\overline{v},b(\overline{v})}$.
\end{itemize}
Using the setting of variable $z_{uv}$ in~\eqref{cons:obj} it follows that $z_{uv}$ is exactly the indicator of edge $\{u,v\}$ being cut by $S$. Thus the objective value in ~\eqref{LP} is $c(\delta S)$.
\hfill \eod
\end{proof}

\begin{lemma}
\eqref{LP} has a polynomial number of variables and constraints. Hence the overall algorithm runs in  polynomial time.
\end{lemma}
\begin{proof}
There are ${n\choose 2}=O(n^2)$ variables $z_{uv}$. Since the tree is binary, we have $|T_i|\le 2d$ for any node $i$, where $d=O(\log n)$ is the depth of the tree-decomposition. Moreover there are only $O(n^2)$ pairs of leaves as there are $O(n)$ leaf nodes. For each pair $\ell_1,\ell_2$ of leaves, we have $|T_{\ell_1}\cup T_{\ell_2}|\le 4d$. Thus $|{\cal P}|\le O(n^2)\cdot 2^{4d}=poly(n)$. By Assumption~\ref{assm:DP}, we have $\max|\fs_{i,\sigma}|=t=O(1)$, so the number of $y$-variables is at most $|{\cal P}|\cdot t^{4d}=poly(n)$. This shows that ~\eqref{LP} has polynomial size and can be solved optimally in polynomial time. Finally, it is easy to see that the rounding algorithm runs in polynomial time.
\hfill \eod
\end{proof}

\begin{lemma}
The algorithm's solution $R$ is always feasible.
\end{lemma}
\begin{proof}
Note that the distribution used in Step~\ref{step:sym-round-sample} is always valid due to Claim~\ref{cl:sym-LP-distr};  so the states $a(i)$s are well-defined.

We now show that  for any node $i\in I$ with children $j,j'$ we have $(a(j),a(j'))\in {\cal F}_{i,a(i)}$. Indeed, at the iteration for node $i$ (when $a(j)$ and $a(j')$ are set) using the probability distribution in~\eqref{eq:sym-SA-round} and by constraint~\eqref{cons:DP}, we obtain  that $(a(j),a(j'))\in {\cal F}_{i,a(i)}$ with probability one.

We show that for each node $i\in I$, the subset $R_i\in \fs_{i,a(i)}$ by induction on the height of $i$. The base case is when $i$ is a leaf. In this case, due to constraint~\eqref{cons:DP-leaf} (and the validity of the rounding algorithm) we know that $\fs_{i,a(i)}\ne \emptyset$. So $R_i=X_{i,a(i)}\in \fs_{i,a(i)}$ by Assumption~\ref{assm:DP}(3). For the inductive step, consider node $i\in I$ with children $j,j'$ where $R_j\in \fs_{j,a(j)}$ and $R_{j'} \in \fs_{j',a(j')}$.  Moreover, from the property above, $(a(j),a(j'))\in {\cal F}_{i,a(i)}$. Now using Assumption~\ref{assm:DP}(3) we have $R_i =X_{i,a(i)} \cup R_{j}\cup R_{j'}\in \fs_{i,a(i)}$. Thus the final solution $R\in \gc$.
\hfill \eod
\end{proof}

\begin{cl}
A vertex $u$ is contained in solution $R$ if and only if $u\in X_{\bar{u},a(\bar{u})}$.
\end{cl}
\begin{proof}
This proof is identical to that of the last property in Claim~\ref{cl:DP-sol-state}.
\hfill \eod
\end{proof}

In the rest of this section, we show that every edge $(u,v)$ is cut by solution $R$ with probability at least $z_{uv}/2$, which would prove the algorithm's approximation ratio. Lemma~\ref{lem:sym-alg-cut1} handles the case when $\bar{u} \in T_{\bar{v}}$ (the case $\bar{v} \in T_{\bar{u}}$ is identical). And Lemma~\ref{lem:sym-alg-cut2} handles the (harder) case when $\bar{u}\not\in T_{\bar{v}}$ and  $\bar{v}\not\in T_{\bar{u}}$.

We first state some useful claims before proving the lemmas.

\begin{observation}[see \cite{GTW13} for a similar use of this principle]
Let $X,Y$ be two jointly distributed $\{0,1\}$ random variables. Then $\Pr(X=1)\Pr(Y=0)+\Pr(X=0)\Pr(Y=1)\ge\frac12[\Pr(X=0,Y=1)+\Pr(X=1,Y=0)]$.\label{obs:joint_indep}
\end{observation}
\begin{proof}
Let $\Pr(X=0,Y=0)=x$, $\Pr(X=0)=a$, $\Pr(Y=0)=b$. The probability table is as below: 
\begin{table}[htb]
\centering
\resizebox{0.5\textwidth}{!}{
\begin{tabular}{ccc}
\cline{1-2}
\multicolumn{1}{|c|}{$\quad x\quad$}   & \multicolumn{1}{c|}{$a-x$}     &$ a$   \\ \cline{1-2}
\multicolumn{1}{|c|}{$\quad b-x\quad$} & \multicolumn{1}{c|}{$\quad1+x-a-b\quad$} & $1-a$ \\ \cline{1-2}
$b$                         & $1-b $                         &    
\end{tabular}
}

\label{tab:pro}
\end{table}

Then we want to show $a(1-b)+b(1-a)\ge\frac12(a+b-2x)\Leftrightarrow a+b+2x\ge4ab$. Since each probability is in $[0,1]$, we have $0\le x\le \min\{a,b\}$ and $a+b-1\le x$.

If $a+b>1$, we have $a+b+2x\ge 3a+3b-2> 1$. If $ab<\frac14$, it is done. If $\frac14\le ab\le 1$, we have $6\sqrt{ab}-2\ge 4ab$. Then $a+b+2x\ge 3a+3b-2\ge 4ab$.

If $a+b\le 1$, we have $a+b+2x\ge a+b\ge 2\sqrt{ab}\ge 4ab$.

Combine the above two cases, we have the observation is true.
\hfill \eod
\end{proof}
\begin{cl}
For any node $i$ and state $s(k)\in \Sigma _k$ for all $k\in T_i$, the rounding algorithm satisfies $\Pr[a(T_i)=s(T_i)]=y(s(T_i))$.\label{cl:path}
\end{cl}
\begin{proof}
We proceed by induction on the depth of node $i$. It is clearly true when $i=r$, i.e. $T_i=\{r\}$. Assuming the statement is true for node $i$, we will prove it for $i$'s children.  Let $j,j'$ be the children nodes of $i$; note that $T_j=T_{j'}=T_i\cup\{j,j'\}$. Then using~\eqref{eq:sym-SA-round}, we have
$$\Pr[a(T_j)=s(T_j) \,\, | \,\, a(T_i)=s(T_i)]\quad =\quad \frac{y(s(T_i\cup\{j,j'\}))}{y(s(T_i))}.$$
Combined with $\Pr[a(T_i)=s(T_i)]=y(s(T_i))$ we obtain  $\Pr[a(T_j)=s(T_j)]=y(s(T_j))$ as desired.\hfill \eod
\end{proof}

\begin{cl}
For any $u,v\in V$, $s(\bar{u})\in \Sigma_{\bar{u}}$ and $s(\bar{v})\in \Sigma_{\bar{v}}$, we have
$$y(s(\{\bar{u},\bar{v}\}))=\sum_{\substack{s(k)\in\Sigma_{k} \\ k \in T_{i}\setminus \bar{u} \setminus \bar{v}}}  y(s(T_i\cup\{\bar{u},\bar{v}\})),$$  where $i$ is the least common ancestor of $\bar{u}$ and $\bar{v}$.\label{cl:leaf}
\end{cl}
\begin{proof}
Since $i$ is the least common ancestor of $\bar{u}$ and $\bar{v}$, we have $T_i\cup\{\bar{u},\bar{v}\}\in{\cal P}$. Then the claim follows by repeatedly applying constraint \eqref{cons:SA}.\hfill \eod
\end{proof}

\begin{lemma}\label{lem:sym-alg-cut1}
Consider any $u,v\in V$ such that $\bar{u} \in T_{\bar{v}}$. Then the probability that edge $(u,v)$ is cut by solution $R$ is $z_{uv}$.
\end{lemma}
\begin{proof}
Applying Claim~\ref{cl:path} with node $i=\bar{v}$, for any $\{s(k)\in \Sigma _k : k\in T_{\bar{v}}\}$, we have $\Pr[a(T_{\bar{v}}) = s(T_{\bar{u}})] = y(s(T_{\bar{u}}))$. Let $D_u=\{s(\bar{u})\in\Sigma_{\bar{u}}|u\in s(\bar{u})\}$ and $D_v=\{s(\bar{v})\in\Sigma_{\bar{v}}|v\in s(\bar{v})\}$. Since $\bar{u}\in T_{\bar{v}}$,
$$\Pr[u\in R, v\not\in R ]=\sum_{s(\bar{u}) \in D_u} \sum_{s(\bar{v})\not\in D_v} \sum_{\substack{s(k)\in\Sigma_{k} \\ k \in T_{\bar{v}}\setminus \bar{u} \setminus \bar{v}}}  y(s(T_{\bar{u}})) = \sum_{s(\bar{u}) \in D_u} \sum_{s(\bar{v})\not\in D_v} y(s(\bar{u},\bar{v}))$$
The last equality above is by repeated application of constraint~\eqref{cons:SA}. Similarly we have
$$\Pr[u\not\in R, v\in R ]= \sum_{s(\bar{u}) \not\in D_u} \sum_{s(\bar{v})\in D_v} y(s(\bar{u},\bar{v})),$$
which combined with constraint~\eqref{cons:r} implies $\Pr[|\{u,v\}\cap R| =1 ] = z_{uv}$.
\hfill \eod \end{proof}

\begin{lemma}\label{lem:sym-alg-cut2}
Consider any $u,v\in V$ such that $\bar{u}\not\in T_{\bar{v}}$ and  $\bar{v}\not\in T_{\bar{u}}$. Then the probability that edge $(u,v)$ is cut by solution $R$ is at least $z_{uv}/2$.
\end{lemma}
\begin{proof}
In order to simplify notation, we define:
\begin{equation*}
z^+_{uv}\,\, =\,\, \sum_{\substack{s(\bar{u})\in\Sigma_{\bar{u}} \\ u\in X_{\bar{u},s(\bar{u})}}}\sum_{\substack{s(\bar{v})\in\Sigma_{\bar{v}} \\ v\not\in X_{\bar{v},s(\bar{v})}}}{y(s(\{\bar{u},\bar{v}\}))},\quad z^-_{uv}\,\, =\,\, \ \sum_{\substack{s(\bar{u})\in\Sigma_{\bar{u}} \\ u\not\in X_{\bar{u},s(\bar{u})}}}\sum_{\substack{s(\bar{v})\in\Sigma_{\bar{v}} \\ v\in X_{\bar{v},s(\bar{v})}}}{y(s(\{\bar{u},\bar{v}\}))}.
\end{equation*}
Note that $z_{uv} = z_{uv}^+ + z_{uv}^-$.

Let $D_u=\{s(\bar{u})\in\Sigma_{\bar{u}}|u\in s(\bar{u})\}$ and $D_v=\{s(\bar{v})\in\Sigma_{\bar{v}}|v\in s(\bar{v})\}$. Let $i$ denote the least common ancestor of nodes $\bar{u}$ and $\bar{v}$. For any choice of states $\{s(k)\in \Sigma_k\}_{k\in T_i}$ define:
$$z^+_{uv}(s(T_i)) = \sum_{s(\bar{u}) \in D_u} \sum_{s(\bar{v})\not\in D_v} \frac{y(s(T_i\cup\{\bar{u},\bar{v}\}))}{y(s(T_i))},$$
and similarly $z^-_{uv}(s(T_i))$.

In the rest of the proof we fix states $\{s(k)\in \Sigma_k\}_{k\in T_i}$ and condition on the event ${\cal E}$ that $a(T_i)=s(T_i)$. We will show:
\begin{equation}\label{eq:sym-cut-prob}
\Pr[|\{u,v\}\cap R|=1 \, |\, {\cal E}] \,\, \ge \,\, \frac12\left(z^+_{uv}(s(T_i))+z^-_{uv}(s(T_i))\right).
\end{equation}

By taking expectation over the conditioning $s(T_i)$, this would imply  Lemma~\ref{lem:sym-alg-cut2}.

We now define the following indicator random variables (conditioned on ${\cal E}$).
\begin{equation*}
I_u=\begin{cases}
0\quad \mbox{ if }a(\bar{u}) \not\in D_u\\
1\quad \mbox{ if }a(\bar{u}) \in D_u
\end{cases} \quad \mbox{and} \quad
I_v=\begin{cases}
0\quad \mbox{ if }a(\bar{v}) \not\in D_v\\
1\quad \mbox{ if }a(\bar{v}) \in D_v
\end{cases}.
\end{equation*}
Observe that $I_u$ and $I_v$ ({\em conditioned} on ${\cal E}$) are independent since $\bar{u}\not\in T_{\bar{v}}$ and $\bar{v}\not\in T_{\bar{u}}$. So,
\begin{equation}\label{eq:sym-alg-cut-prob}
\Pr[|\{u,v\}\cap R|=1 \, |\, {\cal E}] \,\, = \,\, \Pr[I_u=1]\cdot \Pr[I_v=0] +  \Pr[I_u=0]\cdot \Pr[I_v=1]
\end{equation}

For any $s(k)\in\Sigma_k$ for $k\in T_{\bar{u}}\setminus T_i$, we have by Claim \ref{cl:path} and $T_i\sse T_{\bar{u}}$ that
$$\Pr[a(T_{\bar{u}})=s(T_{\bar{u}})\,|\, a(T_i)=s(T_i)] = \frac{\Pr[a(T_{\bar{u}})=s(T_{\bar{u}})]}{\Pr[a(T_i)=s(T_i)]}=\frac{y(s(T_{\bar{u}}))}{y(s(T_i))}.$$
Therefore
$$\Pr[I_u=1]\,\,=\,\,\sum_{s(\bar{u})\in D_u} \,\, \sum_{\substack{k\in T_{\bar{u}}\setminus T_i\setminus \{\bar{u}\} \\s(k)\in\Sigma_k}}\frac{y(s(T_{\bar{u}}))}{y(s(T_i))}=\sum_{s(\bar{u})\in D_u}\frac{y(s(T_{i}\cup\{\bar{u}\}))}{y(s(T_i))}.$$

The last equality follows from the \eqref{cons:SA} constraint. Similarly,
$$\Pr[I_v=1]\,\,=\,\, \sum_{s(\bar{v})\in D_v}\frac{y(s(T_{i}\cup\{\bar{v}\}))}{y(s(T_i))}.$$

Now define $\{0,1\}$ random variables $X$ and $Y$ jointly distributed as:
$$\begin{array}{|c|c|c|}
\hline
&Y=0 & Y=1\\
\hline
X=0 & \Pr[I_v=1] - z_{uv}^-(s(T_i)) & z_{uv}^-(s(T_i))  \\
\hline
X=1 & z_{uv}^+(s(T_i)) & \Pr[I_u=1] - z_{uv}^+(s(T_i))\\
\hline
\end{array}
$$
Note that $\Pr[X=1]=\Pr[I_u=1]$ and $\Pr[Y=1]=\Pr[I_v=1]$. So, applying Observation~\ref{obs:joint_indep} and using~\eqref{eq:sym-alg-cut-prob} we have:
$$\Pr[|\{u,v\}\cap R|=1 \, |\, {\cal E}] \,\, \ge \,\, \frac12\left( \Pr[X=0,Y=1]+\Pr[X=1,Y=0]\right),$$
which implies~\eqref{eq:sym-cut-prob}.
\hfill \eod
\end{proof}
\section{Applications}\label{sec:appln}
In this section, we show a number of graph constraints that satisfy Assumption~\ref{assm:DP} and thereby obtain $\frac12$-approximation algorithms for \gcmc under these constraints (on bounded-treewidth graphs).

Recall that the underlying graph $G$ is given by its tree-decomposition $({\cal T}=(I,F), \{X_i|i\in I\})$ from Theorem~\ref{thm:treedecomp}. Recall also the definition of a dynamic program on this tree-decomposition, as given in Definition~\ref{def:DP}.

\subsection{Independent Set}
Given graph $G=(V,E)$ and edge-weights $\ew:{V\choose 2}\rightarrow \mathbb{R}_+$ we want to maximize $c(\delta S)$ where $S$ is an independent set in $G$.
\smallskip

\noindent For each node $i\in I$ define state space $\Sigma_i = \{\sigma\subseteq X_i \, | \, \sigma\mbox{ is an independent set}\}$. For each node $i\in I$ and $\sigma\in \Sigma_i$, we define:
\begin{itemize}
\item set $X_{i,\sigma} = \sigma$.
\item collection $\fs_{i,\sigma}=\{S\subseteq V_i \, | \, X_i\cap S=\sigma\mbox{ and $S$ is an independent set in }G[V_i]\}$.
\item ${\cal F}_{i,\sigma}=\{(w_{j_1},w_{j_2}) \, |\, \mbox{for each }j\in \{j_1,j_2\},\, w_j\in\Sigma_j\mbox{ such that }w_j\cap X_i=\sigma\cap X_j\}$ which denotes valid combinations. Note that the condition $w_j\cap X_i=\sigma\cap X_j$ enforces $w_j$ to agree with $\sigma$ on vertices of $X_i\cap X_j$.
\end{itemize}

We next show that these satisfy all the conditions in Assumption~\ref{assm:DP}.

\smallskip
\noindent \emph{Assumption~\ref{assm:DP} part 1.} We have $t=\max{|\Sigma_i|}\le 2^k=O(1)$  for bounded-treewidth $k$. Also $p=\max{|{\cal F}_{i,\sigma}|}\le t^2$ since each node has at most two children.

\smallskip
\noindent \emph{Assumption~\ref{assm:DP} part 2.}
By definition, for any  $S\in \fs_{i,\sigma}$ we have $S\cap X_i = \sigma = X_{i,\sigma}$.

\smallskip
\noindent \emph{Assumption~\ref{assm:DP} part 3.} For any leaf $\ell\in I$ and $\sigma\in \Sigma_\ell$ it is clear that $\fs_{\ell,\sigma} = \{X_{i,\sigma}\}$.

Consider now any non-leaf node $i$ and $\sigma\in \Sigma_i$.  Let
\begin{equation} \label{eq:dp-indepset}
{\cal Z}\,\, =\,\, \{X_{i,\sigma}\cup S_{j_1} \cup S_{j_2} \, :\, S_{j_1}\in\fs_{j_1,w_{j_1}}, S_{{j_2}}\in\fs_{j_2,w_{j_2}}, (w_{j_1},w_{j_2}) \in{\cal F}_{i,\sigma}\}.
\end{equation}

We first prove $\fs_{i,\sigma}\subseteq {\cal Z}$. For any $S\in \fs_{i,\sigma}$ and child $j\in \{j_1,j_2\}$  let
$S_j=S\cap V_j$ and $w_j=S\cap X_j$; since $S$ is independent $S_j$ is also an independent set, and $S_j\in \fs_{j,w_j}$. Note that  $S\cap X_i = X_{i,\sigma}$.  Since $V_i=X_i\cup V_{j_1}\cup V_{j_2}$, we have $S=X_{i,\sigma}\cup S_{j_1} \cup S_{j_2}$. Moreover, we have $\sigma\cap X_j = S\cap X_i \cap X_j = w_j\cap X_i$ for each $j\in \{j_1,j_2\}$.  So we have $(w_{j_1},w_{j_2})\in {\cal F}_{i,\sigma}$ and hence $S\in {\cal Z}$.

We next prove ${\cal Z}\subseteq\fs_{i,\sigma}$. Consider any $S=X_{i,\sigma}\cup S_{j_1} \cup S_{j_2}$ as in~\eqref{eq:dp-indepset}. For $j\in\{j_1,j_2\}$ by definition of ${\cal F}_{i,\sigma}$ and $\fs_{j,w_j}$, we have $\sigma\cap X_j = w_j\cap X_i = (S_j\cap X_j)\cap X_i$; since $X_i\cap (V_j \setminus X_j)=\emptyset$ (by definition of the tree-decomposition) we have $X_i\cap S_j = X_i\cap X_j\cap S_j = \sigma\cap X_j$.
 Thus we have $X_i\cap S=\sigma$. It just remains to prove that $S$ is an independent set in $G[V_i]$. Since $S_{j_1}$, $S_{j_2}$ and $X_{i,\sigma}$ are independent sets, if $S$ were not independent then we must have an edge $(u,v)$ where $u\in V_{j_1}\cup X_i$ and $v\in V_{j_2}\setminus  X_i$ (or the symmetric case); this is not possible due to the tree-decomposition. So $S \in \fs_{i,\sigma}$.

\smallskip
\noindent \emph{Assumption~\ref{assm:DP} part 4.} This  follows directly from the definition of $\fs_{i,\sigma}$.

\subsection{Connectivity}

Given graph $G=(V,E)$ and edge-weights $\ew:{V\choose 2}\rightarrow \mathbb{R}_+$ we want to maximize $c(\delta S)$ where $S$ is a connected vertex-set in $G$.
\smallskip

\noindent For each node $i\in I$ define the state space
$$\Sigma_i = \{(B_i,P_i) \, | \, B_i \subseteq X_i,\, P_i \mbox{ is a partition of }B_i\}.$$
Here a state $\sigma=(B_i,P_i)$ specifies which subset $B_i$ of the vertices (in $X_i$) are included in the solution and what is the connectivity pattern $P_i$ among them.

For each node $i\in I$ and $\sigma =(B_i,P_i) \in \Sigma_i$, we define:
\begin{itemize}
\item set $X_{i,\sigma}=B_i$.
\item if $i=r$ (at the root) $\Sigma_r=\{(B_r,P_r)\,|\,B_r\sse X_r,\,P_r=\{B_r\}\}$.
\item if $i\ne r$ then $\fs_{i,\sigma}=\{S\sse V_i \,|\, X_i\cap S=B_i$, each part of $P_i$ is connected in $G[S]$ and every connected component of $G[S]$ contains some vertex of $B_i\}$.
\item partition $\bar{P}_i$ denotes the connected components in $G[B_i]$.
\item ${\cal F}_{i,\sigma}$ consists of $(w_{j_1},w_{j_2})$ where for $j\in\{j_1,j_2\}$, $w_j=(B_j,P_j)\in\Sigma_j$ such that $B_i\cap X_j = B_j\cap X_i$ and each part of $P_j$ contains some vertex of $B_i$,  and $P_i$ is satisfied\footnote{Given two partitions $Q$ and $R$, their union $P=Q\cup R$ is the refined partition where a pair of elements are in the same part iff they occur in the same part of either $Q$ or $R$. Moreover, a partition $P$ is said to be satisfied by another partition $P'$ if $P'$ is a refinement of $P$, i.e. every pair of elements in the same part of $P$ also lie in the same part of $P'$.} by $\bar{P}_i\cup P_{j_1}\cup P_{j_2}$. Note that for some states there may be no such pair $(w_{j_1},w_{j_2})$ : in this case ${\cal F}_{i,\sigma}$ is empty.
\end{itemize}

\noindent \emph{Assumption~\ref{assm:DP} part 1.} For each node $i$, the possible number of vertex subsets $B_i$ is at most $2^k$ and the possible number of partitions $P_i$ is at most $k^k$, where $k$ is the treewidth. So for a bounded-treewidth $k$, we have $t=\max{|\Sigma_i|}\le k^{k+1}=O(1)$. Then $p=\max|{\cal F}_{i,\sigma}|\le t^2=O(1)$.

\smallskip
\noindent \emph{Assumption~\ref{assm:DP} part 2.} This follows directly from the definition of $\fs_{i,\sigma}$ and $X_{i,\sigma}$.

\smallskip
\noindent \emph{Assumption~\ref{assm:DP} part 3.} Let ${\cal Z}$ be as in~\eqref{eq:dp-indepset} with the new definitions of $\fs$ and ${\cal F}$ for connectivity (as above). The leaf case is trivial, so we consider a non-leaf node $i\in I$ and $\sigma=(B_i,P_i)\in \Sigma_i$. To reduce notation we just use $j$ to denote a child of $i$; we will not specify $j\in \{j_1,j_2\}$ each time.

We first prove $\fs_{i,\sigma}\subseteq {\cal Z}$. For any $S\in \fs_{i,\sigma}$, let $S_j=V_j\cap S$ and $B_j=X_j\cap S$. Let $P_j$ be a partition of $B_j$ with a part $C\cap B_j$ for every connected component $C$ in $G[S_j]$. Let $w_j=(B_j,P_j)$. We will show that $S_j\in\fs_{j,w_j}$ and $(w_{j_1},w_{j_2})\in{\cal F}_{i,\sigma}$.

\begin{itemize}
\item $S_j\in\fs_{j,w_j}$. By definition of $B_j$, we have $X_j\cap S_j=X_j\cap V_j\cap S=X_j\cap S=B_j$. We only need to prove each connected component of $G[S_j]$ has at least one vertex of $B_j$. We will in fact show that each component of $G[S_j]$ has at least one vertex of $B_i$ (i.e. in $B_j\cap B_i$).  Suppose (for contradiction) there is some connected component $C$ in $G[S_j]$ which does not have any vertex of $B_i$. By $S\in \fs_{i,\sigma}$ we know that in the (larger) graph $G[S]$ component $C$ has to be connected to some vertex $u\in B_i$. Then there is a path $\pi$ in $G[S]$ from some vertex $u'\in C$ to $u$ such that $u'$ is the {\em only} vertex of $C$ on $\pi$ (see also Figure \ref{fig:connect}). Let $(u',v')$ be the first edge of $\pi$, so $u'\in C\sse S_j$ and $v'\in S\setminus S_j$.  By tree-decomposition, there is some node containing both $u'$ and $v'$. Since $u',v'\in V_i$ and $u'\in S_j,v'\not\in S_j$, that node can only be $i$. This means $u'\in B_i$, contrary to our assumption. Therefore we have $S_j\in\fs_{j,w_j}$.

\item $(w_{j_1},w_{j_2})\in{\cal F}_{i,\sigma}$.  We have $B_i\cap X_j = B_j\cap X_i$ by definition of $w_j$. Since we already proved that each connected component of $G[S_j]$ has at least one vertex of $B_j\cap B_i$, we know that each part of partition $P_j$ has at least one vertex of $B_i$. By tree-decomposition we have $G[V_i]=G[X_i]\cup G[V_{j_1}]\cup G[V_{j_2}]$, so $G[S]=G[B_i]\cup G[S_{j_1}]\cup G[S_{j_2}]$. Hence partition $P_i$ is satisfied by $\bar{P}_i\cup P_{j_1}\cup P_{j_2}$. Thus $(w_{j_1},w_{j_2})\in{\cal F}_{i,\sigma}$.
\end{itemize}

\begin{figure}[htb]
  \centering
  \includegraphics[height=0.35\textwidth]{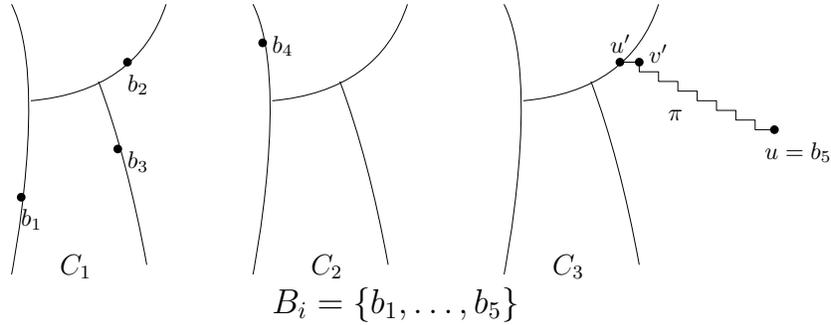}
\caption{Maximal connected component of $G[S_j]$}
\label{fig:connect}
\end{figure}

Next we prove ${\cal Z}\sse\fs_{i,\sigma}$. Consider any $S\in {\cal Z}$ given by $S=B_i\cup S_{j_1}\cup S_{j_2}$ as in~\eqref{eq:dp-indepset}. The fact that $S\cap X_i=B_i$ follows exactly as in the case of an independent-set constraint. Since $P_i$ is satisfied by $\bar{P}_i\cup P_{j_1}\cup P_{j_2}$ and $S_j$ connects up each part of $P_j$, it follows that $G[S] = G[B_i]\cup G[S_{j_1}]\cup G[S_{j_2}]$ connects up each part of $P_i$. It remains to show that each connected component of $G[S]$ has a vertex of $B_i$. Since
$(w_{j_1},w_{j_2})\in{\cal F}_{i,\sigma}$ we know that each part of $P_j$ has a $B_i$-vertex. By $S_j\in \fs_{j,w_j}$, we know that each component of $G[S_j]$ contains some vertex $u\in B_j$, and  this vertex $u$ is connected to some vertex $v\in B_i$ (as each part of $P_j$ contains a $B_i$-vertex); so every component of $G[S_j]$ contains some vertex of $B_i$.  Hence each component of $G[S]= G[B_i]\cup G[S_{j_1}]\cup G[S_{j_2}]$ also contains some vertex of $B_i$. 

\smallskip
\noindent \emph{Assumption~\ref{assm:DP} part 4.} By our definition of $\Sigma_r$, any solution given by $\fs_{r,\sigma}$ requires all chosen vertices to be connected. Thus this assumption is satisfied.

\subsection{Vertex Cover}

Given graph $G=(V,E)$ and edge-weights $\ew:{V\choose 2}\rightarrow \mathbb{R}_+$ we want to maximize $c(\delta S)$ where $S$ is a vertex cover in $G$ (i.e. $S$ contains at least one end-point of each edge in $E$).
\smallskip

\noindent For each node $i\in I$ define the state space $\Sigma_i = \{\sigma\subseteq X_i \, | \, \sigma\mbox{ is a vertex cover in }G[X_i]\}$. For each node $i\in I$ and $\sigma\in \Sigma_i$, we define:
\begin{itemize}
\item set $X_{i,\sigma} = \sigma$.
\item collection $\fs_{i,\sigma}=\{S\subseteq V_i \, | \, X_i\cap S=\sigma\mbox{ and $S$ is a vertex cover in }G[V_i]\}$.
\item ${\cal F}_{i,\sigma}=\{(w_{j_1},w_{j_2}) \, |\, \mbox{for each }j\in \{j_1,j_2\},\, w_j\in\Sigma_j\mbox{ such that }w_j\cap X_i=\sigma\cap X_j\}$ which denotes valid combinations. Note that the condition $w_j\cap X_i=\sigma\cap X_j$ enforces $w_j$ to agree with $\sigma$ on vertices of $X_i\cap X_j$.

\end{itemize}

The proof of the above notation satisfying Assumption~\ref{assm:DP} is identical to the independent set proof.

\subsection{Dominating Set}

Given graph $G=(V,E)$ and edge-weights $\ew:{V\choose 2}\rightarrow \mathbb{R}_+$ we want to maximize $c(\delta S)$ where $S$ is a dominating set in $G$ (i.e. every vertex in $V$ is either in $S$ or a neighbor of some vertex in $S$).
\smallskip

\noindent For each node $i\in I$ define the state space $$\Sigma_i=\{(B_i,Y_i)\,|\,B_i\sse X_i,\,Y_i\sse X_i\}.$$ For vertex set $S\sse V$, we use $N(S)$ to denote $S$ and the neighbor vertices of $S$.

Here a state $\sigma=(B_i,Y_i)$ specifies which subset $B_i$ of the vertices (in $X_i$) are included in the solution and what subset  $Y_i$ of the vertices (in $X_i$) should be dominated. 

For each node $i\in I$ and $\sigma=(B_i,Y_i)\in\Sigma_i$, we define:
\begin{itemize}
\item set $X_{i,\sigma}=B_i$
\item if $i=r$ (at the root) $\Sigma_r=\{(B_r,Y_r)\,|\,B_r\sse X_r,\,Y_r=\emptyset\}$.
\item if $i\ne r$ then $\fs_{i,\sigma}=\{S\sse V_i|\,X_i\cap S=B_i, S$ is a dominating set of $V_i\setminus Y_i$ in $G[V_i]\}$
\item ${\cal F}_{i,\sigma}$ consists of $(w_{j_1},\,w_{j_2})$ where for $j\in\{j_1,\,j_2\}$, $w_j=(B_j,Y_j)\in\Sigma_j$ such that $B_i\cap X_j=B_j\cap X_i$ and $V_i\setminus Y_i \sse(V_{j_1}\setminus Y_{j_1})\cup(V_{j_2}\setminus Y_{j_2})\cup N(B_i)$.
Note that for some states there may be no such pair $(w_{j_1},\,w_{j_2})$ : in this case ${\cal F}_{i,\sigma}$ is empty.
\end{itemize}

\noindent \emph{Assumption~\ref{assm:DP} part 1.} For each node $i$, the possible number of vertex subsets $X_i$, $Y_i$ is at most $2^k$, where $k$ is the treewidth. So for a bounded-treewidth $k$, we have $t=\max{|\Sigma_i|}\le 2^{2k}=O(1)$. Then $p=\max|{\cal F}_{i,\sigma}|\le t^2=O(1)$.

\smallskip
\noindent \emph{Assumption~\ref{assm:DP} part 2.} This follows directly from the definition of $\fs_{i,\sigma}$ and $X_{i,\sigma}$.

\smallskip
\noindent \emph{Assumption~\ref{assm:DP} part 3.} Let ${\cal Z}$ be as in~\eqref{eq:dp-indepset} with the new definitions of $\fs$ and ${\cal F}$ for dominate set (as above). The leaf case is trivial, so we consider a non-leaf node $i\in I$ and $\sigma=(B_i,Y_i)\in \Sigma_i$. To reduce notation we just use $j$ to denote a child of $i$; we will not specify $j\in \{j_1,j_2\}$ each time.

We first prove $\fs_{i,\sigma}\subseteq {\cal Z}$. For any $S\in \fs_{i,\sigma}$, let $S_j=V_j\cap S$ and $B_j=X_j\cap S$. Let $Y_j=X_j\setminus N(S_j)$. Let $w_j=(B_j,Y_j)$. We will show that $S_j\in\fs_{j,w_j}$ and $(w_{j_1},w_{j_2})\in{\cal F}_{i,\sigma}$.

\begin{itemize}
\item $S_j\in\fs_{j,w_j}$. By definition of $B_j$, we have $X_j\cap S_j=X_j\cap V_j\cap S=X_j\cap S=B_j$. We only need to prove $S_j$ is a dominating set of $V_j\setminus Y_j$ in $G[V_j]$. For all $v\in V_j\setminus Y_j$: If $v\in X_i$, since $v\in V_j$, we have $v\in X_j$. Since $v\not\in Y_j$ and $v\in X_j$, by $Y_j=X_j\setminus N(S_j)$, we have $v\in N(S_j)$ by tree-decomposition, that is $v$ is dominated by $S_j$. If $v\not\in X_i$, we have $v\in V_i\setminus Y_i$. Then $v$ is dominated by $S$. There is some $u\in S$ such that $(u,v)\in E$. Then by tree-decomposition, since $v\in V_j$ and $v\not\in X_i$, we have $u\in V_j$. Then since $S_j=S\cap V_j$, we have $u\in S_j$. $v$ is dominated by $S_j$. Then we have $S_j$ will dominate $V_j\setminus Y_j$. Thus we have $S_j\in\fs_{j,w_j}$.

\item $(w_{j_1},w_{j_2})\in{\cal F}_{i,\sigma}$.  We have $B_i\cap X_j = B_j\cap X_i$ and $Y_j=X_j\cap(Y_i\cup N(B_j))$ by definition of $w_j$. It remains to show that $V_i\setminus Y_i \sse(V_{j_1}\setminus Y_{j_1})\cup(V_{j_2}\setminus Y_{j_2})\cup N(B_i)$. For all $v\in V_i\setminus Y_i$, we have $v$ is dominated by $S$. There is $u\in S$ such that $(u,v)\in E$. If $u\in X_i$, then we have $u\in B_i$ thus $v\in N(B_i)$. If $u\in V_j\setminus X_i$, we have $u\in S_j$. By tree-decomposition, $u\not\in X_i$, $u\in V_j$ and $(u,v)\in E$ gives us $v\in V_j$. If $v\not\in X_j$, we have $v\in V_j\setminus Y_j$. If $v\in X_j$, since $u\in S_j$, $v\in N(u)$, we have $v\in N(S_j)$. Then by $Y_j=X_j\setminus N(S_j)$, we have $v\not\in Y_j$. Thus $v\in V_j\setminus Y_j$. Therefore, for all $v\in V_i\setminus Y_i$, we have $v\in (V_{j_1}\setminus Y_{j_1})\cup(V_{j_2}\setminus Y_{j_2})\cup N(B_i)$. Thus $V_i\setminus Y_i\sse (V_{j_1}\setminus Y_{j_1})\cup(V_{j_2}\setminus Y_{j_2})\cup N(B_i)$.
\end{itemize}

Next we prove ${\cal Z}\sse\fs_{i,\sigma}$. Consider any $S\in {\cal Z}$ given by $S=B_i\cup S_{j_1}\cup S_{j_2}$ as in~\eqref{eq:dp-indepset}. The fact that $S\cap X_i=B_i$ follows exactly as in the case of an independent-set constraint. It remains to show that $S$ is a dominating set of $V_i\setminus Y_i$. For all $v\in V_i\setminus Y_i$, we have $v\in (V_{j_1}\setminus Y_{j_1})\cup(V_{j_2}\setminus Y_{j_2})\cup N(B_i)$. Since $S_{j}$ is a dominate set of $V_{j}\setminus Y_{j}$ and $B_i$ is a dominate set of $N(B_i)$, we have $v$ is dominated by $B_i\cup S_{j_1}\cup S_{j_2}$, $v$ is dominated by $S$. Thus we have $S$ is a dominating set of $V_i\setminus Y_i$. $S\in \fs_{i,\sigma}$.

\smallskip
\noindent \emph{Assumption~\ref{assm:DP} part 4.} By our definition of $\Sigma_r$, any solution given by $\fs_{r,\sigma}$ requires all vertices are dominated. Thus this assumption is satisfied.

\section{Bounded-genus and Excluded-minor Graphs}
Here we use known decomposition results to show that our results can be extended to a larger class of graphs, and prove Corollary~\ref{cor:em} and~\ref{cor:bg}.
\subsection{Excluded-minor graph}
Recall the following decomposition of any excluded-minor graph into graphs of bounded treewidth.
\begin{theorem}{\cite{DHK11}}
For a fixed graph $H$, there exists a constant $c_H$ such that, for any integer $h\ge 1$ and for every $H$-minor-free graph $G$, the vertices of $G$ can be partitioned into $h+1$ sets such that any $h$ of that sets induce a graph of treewidth at most $c_Hh$. Furthermore, such partition can be found in polynomial time.\label{thm:em}
\end{theorem}
Algorithm \ref{alg:em} for Corollary \ref{cor:em} is given below. For a subset $V_i\sse V$, let  $G_i$ be the graph obtained by contracting $V_i$   to $v_{new}$. Then the edge weight on $G_i$ is defined as 
\begin{equation}
c_i(u,v)=\begin{cases}c(u,v),\mbox{ if $u,v\in V\setminus V_i$}\\\sum_{w\in V_i}c(u,w),\mbox{ if $u\in V\setminus V_i,v=v_{new}$}\end{cases}\label{eq:ci}
\end{equation} 
We have $v_{new}$ can increase treewidth by at most one since we can add it to each tree node and give a feasible tree-decomposition.

\begin{algorithm}[t] 
\label{alg:em}
\LinesNumbered
\SetKwInOut{Input}{Input}\SetKwInOut{Output}{Output}
\SetKwFor{Do}{Do}{\string:}{end}
 \Input{$H$-minor-free graph $G$}
 \Output{A vertex set in ${\cal C}_G$}
Use Theorem~\ref{thm:em} to partition $V$ into $V_1,\dots,V_h$\;

\If{independent-set constraint}{
\For{$i=1$ to $h$}{
Solve \gcmc in $G[V\setminus V_i]$ with edge-weight $c$\;
Let $S_i$ be the solution.\;
$S=\argmax_{S\in\{S_i\}}{c(\delta S_i)}$}}
\If{vertex-cover or dominating-set constraint}{
\For{$i=1$ to $h$}{
Solve \gcmc in $G_i$ with edge-weight $c_i$ and require $v_{new}$ to be part of the solution\;
\tcc{this requirement can be achieved by adding constraints $y(s(r))=0$ for all $v_{new}\not\in s(r)$ to \eqref{LP}}
Let $S'_i$ be the solution.\;
$S_i=S'_i\setminus\{v_{new}\}\cup V_i$\;
$i=\argmax_{j=1\dots h}{c(\delta (S_j))}$\;
$S=S_i$\;
}
}
\Return S\;
\caption{Algorithm for excluded minor graph}
\end{algorithm}

We will show Corollary \ref{cor:em} with the following claims.
\begin{cl}
Let $V_1,\dots,V_h$ be a partition of $V$. Let $S'$ be any vertex subset of $V$. Then there is some $i$ such that  $c(\delta (S'\setminus V_i))\ge(1-\frac2h)c(\delta S')$.\label{cl:vi}
\end{cl}
\begin{proof}
Since $V_i,\dots, V_h$ is a partition of $V$, we have $\sum_{i=1}^h{c(\delta (S'\cap V_i))}\le 2c(\delta S')$. Then $\min_{i}{c(\delta (S'\cap V_i))}\le \frac{2}{h}c(\delta S')\Leftrightarrow\max_{i}{c(\delta (S'\setminus V_i))}\ge(1-\frac{2}{h})c(\delta S')$.\label{cl:si}\hfill \eod
\end{proof}
\begin{cl}
Let $V_1,\dots,V_h$ be a partition of $V$. Let $S'$ be any vertex subset of $V$. Then there is some $i$ such that  $c(\delta (S'\cup V_i))\ge(1-\frac2h)c(\delta S')$.\label{cl:vi2}
\end{cl}
\begin{proof}
Since $V_i,\dots, V_h$ is a partition of $V$, we have $\sum_{i=1}^h{(c(\delta S') - c(\delta (S'\cup V_i)))}\le 2c(\delta S')$. Then $\min_{i}{(c(\delta S') - c(\delta (S'\cup V_i)))} \le \frac{2}{h}c(\delta S')\Leftrightarrow\max_{i}{c(\delta (S'\cup V_i))}\ge(1-\frac{2}{h})c(\delta S')$.\label{cl:si}\hfill \eod
\end{proof}
\begin{cl}
Let $V_i$ be a subset of $V$. Suppose $S'$ is a feasible solution of some \gcmc problem and $S\setminus V_i$ is a feasible solution in $G[V\setminus V_i]$ with same graph constraint. Then the solution $S_i$ given by \gcmc algorithm in $G[V\setminus V_i]$ has cut value $c(\delta S_i)\ge\frac12c(\delta S'\setminus V_i)$.\label{cl:si}
\end{cl}
\begin{proof}
Since $S\setminus V_i$ is a feasible solution in $G[V\setminus V_i]$, we have the optimal solution, $S^*$ in $G[V\setminus V_i]$ has cut value $c(\delta S^*)\ge c(\delta S'\setminus V_i)$. By Theorem \ref{thm:sym}, we have $c(\delta S_i)\ge \frac12c(\delta S^*)$, thus we have $c(\delta S_i)\ge\frac12c(\delta S'\setminus V_i)$.\hfill \eod
\end{proof}
\begin{cl}
Let $V_i$ be a subset of $V$. Let $G_i$ and $c_i$ be defined as \eqref{eq:ci}. Suppose $S'$ is a feasible solution of some \gcmc problem and $S'\setminus V_i\cup\{v_{new}\}$ is a feasible solution in $G_i$ with same graph constraint. Then the solution $S_i$ given by \gcmc algorithm in $G_i$ has cut value $c_i(\delta S_i)\ge\frac12c(\delta (S'\cup V_i))$ and $c(\delta (S_i\setminus\{v_{new}\}\cup V_i))=c_i(\delta S_i)$.\label{cl:si2}
\end{cl}
\begin{proof}
Since $S'\setminus V_i\cup\{v_{new}\}$ is a feasible solution in $G_i$, we have the optimal solution, $S^*$ in $G_i$ has cut value $c_i(\delta S^*)\ge c_i(\delta (S'\setminus V_i\cup\{v_{new}\}))$. By Theorem \ref{thm:sym}, we have $c_i(\delta S_i)\ge \frac12c_i(\delta S^*)$. Then by definition of $c_i$ we have $c_i(\delta S_i)\ge\frac12c(\delta (S'\cup V_i\}))$ and $c(\delta (S_i\setminus\ \{v_{new}\}\cup V_i))=c_i(\delta S_i)$.\hfill \eod
\end{proof}
\begin{cl}
Let $S^*$ be the optimal solution of \gcmc with independent-set constraint. The solution $S$ given by the Algorithm \ref{alg:em} is a feasible solution to the original \gcmc problem and $c(\delta S)\ge \frac12(1-\frac{2}{h})c(\delta S^*)$.
\end{cl}
\begin{proof}
Since $S$ is an independent set in $G[V\setminus V_i]$, then it is an independent set in $G$. $S$ is a feasible solution to the original \gcmc problem. Also we have  that $S^*\setminus V_i$ is a feasible solution in $G[V\setminus V_i]$. Then by Claim \ref{cl:si}, we have $c(\delta S)\ge \frac12 c(\delta (S^*\setminus V_i))$ and by Claim \ref{cl:vi}, we have $c(\delta (S^*\setminus V_i))\ge(1-\frac{2}{h})c(\delta S^*).$ Combine the last two inequalities, we have $c(\delta S)\ge \frac12(1-\frac{2}{h})c(\delta S^*)$.
\hfill \eod
\end{proof}
\begin{cl}
Let $S^*$ be the optimal solution of \gcmc with vertex-cover constraint. The solution $S$ given by the Algorithm \ref{alg:em} is a feasible solution to the original \gcmc problem and $c(\delta S)\ge \frac12(1-\frac{6}{h})c(\delta S^*)$.
\end{cl}
\begin{proof}
Since $S'_i$ is a vertex cover of $G_i$. Then by definition of $G_i$, $S_i\cup V_i$ is a vertex cover of $G$. And by definition of $G_i$ and $c_i$, we have $S^*\cup\{v_{new}\}\setminus V_i$ is a feasible solution for $G_i$. Then by Claim \ref{cl:si2}, $c(\delta S)\ge\frac12c(\delta S^*\cup V_i)$. And by Claim \ref{cl:vi2}, we have $c(\delta S^*\cup V_i)\ge(1-\frac2h)c(\delta S^*)$. Combine the last two inequalities, we have $c(\delta S)\ge \frac12(1-\frac{2}{h})c(\delta S^*)$.\hfill \eod
\end{proof}
\begin{cl}
Let $S^*$ be the optimal solution of \gcmc with dominating-set constraint. The solution $S$ given by the Algorithm \ref{alg:em} is a feasible solution to the original \gcmc problem and $c(\delta S)\ge \frac12(1-\frac{2}{h})c(\delta S^*)$.
\end{cl}
\begin{proof}
Since $S'_i$ is a dominating set of $G_i$. Then by definition of $G_i$, $S_i\cup V_i$ is a dominating set of $G$. And by definition of $G_i$ and $c_i$, we have $S^*\cup\{v_{new}\}\setminus V_i$ is a feasible solution for $G_i$. Then by Claim \ref{cl:si2}, $c(\delta S)\ge\frac12c(\delta S^*\cup V_i)$. And by Claim \ref{cl:vi2}, we have $c(\delta S^*\cup V_i)\ge(1-\frac2h)c(\delta S^*)$. Combine the last two inequalities, we have $c(\delta S)\ge \frac12(1-\frac{2}{h})c(\delta S^*)$.
\hfill \eod
\end{proof}
\subsection{Bounded-genus graph}
Here we use:
\begin{theorem}{\cite{DHK05}}
For a bounded-genus graph $G$ and an integer $h$, the edge of $G$ can be partitioned in $h$ color classes $E_1,\dots,E_h$ such that contracting all the edges in any color class leads to a graph with treewidth $O(h)$. Further, the color classes are obtained by a radial coloring and have the following property: If edge $e=(u,v)$ is in class $i$, then every edge $e'$ such that $e\cap e'\ne\emptyset$ is in class $i-1$ or $i$ or $i+1$.\label{thm:bg}
\end{theorem}
The proof of Corollary \ref{cor:bg} using Theorem~\ref{thm:bg} is identical to the proof in \cite{HKMPS15} for the ``uniform'' connected max-cut problem. For each $E_i$, let $S_i$ be the solution in $G$ with $E_i$ contracted,  then $S'_i=\{v|v\in S_i \mbox{ or } v$ is contracted to some vertex of $S_i\}$ is connected in $G$.  Although our edge-weights $c$ are defined on a complete graph, the proof of \cite{HKMPS15} still works. The main reason is that each vertex gets contracted in at most 3 of the graphs $G$ with $E_i$ contracted.
     
\bibliographystyle{splncs03}
\bibliography{gcmcbib}

\end{document}